\documentclass{llncs}

\usepackage{amssymb}
\usepackage{amsmath}

\newcommand{\N}{\mathbb{N}}

\newcommand{\fs}{2^{<\omega}}
\newcommand{\K}{\ensuremath{\mathrm{K}} }
\newcommand{\C}{\ensuremath{\mathrm{C}} }
\newcommand{\A}{\ensuremath{\mathrm{A}} }
\newcommand{\B}{\ensuremath{\mathrm{B}} }
\newcommand{\rem}{\noindent \textbf{Remark:} }
\newcommand{\remk}[1]{\noindent \textbf{Remark #1:} }

\title{Towards an axiomatic system for Kolmogorov complexity} 
\author{Antoine Taveneaux}  \institute{ LIAFA, CNRS \& Universit\'e de Paris 7 }

\begin{document}

\maketitle

\begin{abstract}
In~\cite{shenpapier82}, it is shown that four of its basic functional properties are enough to characterize plain Kolmogorov complexity, hence obtaining an axiomatic characterization of this notion. In this paper, we try to extend this work, both by looking at alternative axiomatic systems for plain complexity and by considering potential axiomatic systems for other types of complexity. First we show that the axiomatic system given by Shen cannot be weakened (at least in any natural way). We then give an analogue of Shen's axiomatic system for conditional complexity. 
In the second part of the paper, we look at prefix-free complexity and try to construct an axiomatic system for it.  We show however that the natural analogues of Shen's axiomatic systems fail to characterize prefix-free complexity.

\end{abstract}

\section{Introduction}

The concept of Kolmogorov complexity was introduced independently by Kolmogorov (in \cite{kolmo65}) and Chaitin (in \cite{papierchaitin}). The aim of Kolmogorov complexity is to quantify the amount of ``information'' contained in finite objects, such as binary strings. This idea can be used to give an answer to the philosophical question, ``What does it mean for a single object to be random?'' 

The usual definition of Kolmogorov complexity uses the existence of an optimal Turing machine. However, it is not immediate from that definition that Kolmogorov complexity is satisfactory as a measure of information. One is only convinced after deriving certain fundamental facts about it, such as: most strings have maximal complexity, the complexity of a pair $(x,y)$ is not (much) greater than the sum of the complexities of~$x$ and $y$, etc. Therefore, a natural question is to ask whether there exists an axiomatic system characterizing Kolmogorov complexity uniquely via some of its functional properties. And of course, as with any axiomatic system, we want the axiomatic system to be minimal, i.e.\ to contain no superfluous axiom.

 Such a characterization was given by Shen in \cite{shenpapier82} for plain complexity. In Section~\ref{sectionplain} we recall this characterization and adapt it to provide an axiomatic system for conditional complexity. We then study whether we can weaken the hypotheses of this characterization of plain complexity in a natural way and show that it is indeed not possible. In particular, one of the hypotheses in the theorem states that applying a partial computable function to a string does not increase its Kolmogorov complexity (up to an additive constant), and we show that this hypothesis cannot be restricted to total computable functions. To show that we need the power of  partial computable functions to characterize plain complexity, we introduce a notion of complexity for functions that are total on a initial segment of the integers; this notion of complexity is robust under the application of total computable functions, but differs from Kolmogorov complexity. 
 
A second natural question would be, ``Is there a similar axiomatic system for prefix-free Kolmogorov complexity?" Unlike the plain complexity case, we show that the classical properties of prefix-free complexity are not sufficient to characterize it. Since prefix-free complexity is greater than plain complexity, we have to choose a larger upper bound and a tighter lower bound to characterize {\K} (where $\K (x)$ is prefix-free complexity of the string $x$, see below). Actually, all basic upper bounds on prefix-free complexity fail to characterize it. To show that our the classical properties of prefix-free complexity do not characterize it, we construct a counter-example defined by  $\A = \K + f $ with~$f$ a very slow growing function. To build such a slow function, we define an operator that slows down sub-linear non-decreasing functions while preserving their computational properties (computability or semicomputability). 

Throughout the paper we will identify natural numbers and finite strings in a natural way (the set of finite strings is denoted by $\fs$). We denote by~$\log (x)$ the discrete binary logarithm of $x$. We fix an effective enumeration of the Turing machines and we denote by $\psi_e$ the function computed by the $e^\text{th}$ machine. For each machine~$T$ and string~$x$, the complexity of $x$ relatively to $T$ is: $$\C_T (x )= \min \{n \, | \, \exists y \in \fs \text{ such that }  |y |=n \text{ and } T(y )= x \} $$ and throughout the paper we fix an optimal machine~$\mathbb{U}$ (i.e.\ a machine such that for all machines $T$ we have $\C_{\mathbb{U}} \leq \C_T + O(1)$, see \cite{Nies_2009} for a existence proof of a such a machine) and set $\C = \C_{\mathbb{U}}$. In the same way, we fix an optimal prefix-free machine $\mathbb{U}'$  (i.e. a machine with prefix-free domain) and set $\K = \C_{\mathbb{U}'}$. $\C (x) $ and $\K (x)$ denote the plain complexity and prefix-free complexity of~$x$, respectively. 

Conditional Kolmogorov complexity is an extension of the above notions which quantifies the information of a string~$x$ \emph{relative} to another string~$y$. More precisely, the complexity of $x$ given $y$, relative to the machine $T$, is: $$\C_T (x|y) =\min \{n \, | \, \exists z \in \fs \text{ such that }  |z |=n \text{ and } T(\langle z,y \rangle )= x \} $$
 As above we can define $\C (.| .) = \C_{\mathbb{U}} (. |. )$ and $\K (.| .) = \C_{\mathbb{U}'} (. |. )$.

\section{Plain complexity}\label{sectionplain}

As mentioned above, Shen showed in~\cite{shenpapier82} that four basic properties are sufficient to fully characterize plain Kolmogorov complexity:
\begin{enumerate}

\item Upper-semicomputability: $\C$ is not computable but it is upper semicomputable  (i.e.\ the predicate $\C (x) \leq k$ is uniformly computably enumerable in $x$ and $k$). 

\item Stability: a recursive function cannot increase the complexity of a string by more than an additive constant. 

\item Explicit description: the length of the smallest description of a string (i.e. its plain complexity) is not much bigger than the string itself. 

\item Counting:  no more than  $2^{n}$ of the strings have a complexity less than $n$.
\end{enumerate}

Formally, Shen's theorem states the following. 

\begin{theorem}\cite{shenpapier82} \label{shenC}
Let $\A: \fs \rightarrow \N$ be some function. Suppose \A satisfies the following four properties:
\begin{enumerate}

\item \label{hypshen1} \A is upper semi-computable.

\item \label{hypshen2} For every partial computable function $f: \fs \rightarrow \fs $  there exists a constant $c_f$ such that for each $\A (f(x))\leq \A (x) +c_f $ for each $x \in \fs$.

\item \label{hypshen3}  $ \A (x)\leq |x |  +O(1) $ for all $x\in \fs$. 
\item \label{hypshen4} $|\{ x | \A (x) \leq n \}|=O(2^{n})$.
\end{enumerate}

Then $\A (x) = \C (x) +O(1)$.

\end{theorem}

\begin{proof}
We give a quick sketch of the proof.

To show  $\A \leq \C +O(1)$, let $x^*$ denote the shortest description of $x$ (for the complexity \C). By hypotheses \ref{hypshen2} and \ref{hypshen3} we have: $$\A (x) = \A (\mathbb{U} (x^*)) \leq \A (x^*) +O(1) \leq |x^*| + O(1) = \C (x) + O(1).$$

To show $\C \leq \A +O(1)$, we consider~$y$ and~$n$ such that $\A(y) = n$. Since \A is upper semi-computable, the set $\{ x | \A (x) \leq n \}$ is uniformly computably enumerable. Since there exists a uniform~$d$ such that $|\{ x | \A (x) \leq n \}|=2^{n+d}$,  we can describe $y$ with only $n+d$ bits (this description $\overline{y}$ such that $|\overline{y} | =n+d$ represents the rank of~$y$ in an enumeration of $\{ x | \A (x) \leq |\overline{y} | - d \}$). So, for all~$y$ we have $ \C (y) \leq n +d +O(1) = \A (y) +O(1)$. \qed

\end{proof}

\remk{1} The authors of \cite{shenbouquin} show that conditions~\ref{hypshen3} and~\ref{hypshen4} can be replaced by ``There exists a constant $c$ such that $|\{ x | \A (x) \leq n \}|\in[2^{n-c},2^{n+c}] $." We can also replace conditions \ref{hypshen2} and \ref{hypshen3} by ``For every partial computable function~$f$ there exists a constant $c_f$ such that $\A (f(x))\leq | x | +c_f $ for each $x \in \fs$." Finally condition \ref{hypshen4} can be replaced by the stronger version ``$|\{ x | \A (x) \leq n -k \}|=O(2^{n-k})$." 

\remk{2} With essentially the same proof, one can show a similar result for conditional plain complexity. The following system  characterizes of conditional plain complexity:

\begin{enumerate}

\item[\textbullet] Uniformly in $  x, y \in \fs$, $\B (x|y) $ is computable from above.
\item[\textbullet] For all $  x, y \in \fs$, $ \B (x|y)\leq |x |  + O(1) $. 
\item[\textbullet] For each $  y \in \fs$ we have $|\{ x | \B (x|y) \leq n \}|=O(2^n)$ (such that $O(2^n)$ do not depend of $y$).
\item[\textbullet] For all $y$ and for every partial computable function $f$ from $\fs $ to $\fs $ there exists a constant $c_f$ such that for each $x \in \fs$: $$\B (f( x  )|y)\leq \B (x |y ) +c_f .$$ 

\end{enumerate}

To characterize the conditional aspect, we add to the four previous items the hypothesis $\text{``}\B (\langle x,y \rangle |y)\leq \B (x |y ) + O(1) \text{"}$.
Note however that replacing this last condition by $\B (x|x) = O(1)$ would not be sufficient.% to characterize the plain conditional complexity.  We give the complete proof of these results in the appendix. 

\subsection{Weakening the  hypotheses}

Shen's theorem raises a natural question:  Are all 4 conditions actually needed? 
In this subsection we discuss this question. First, it is not hard to see that none of the hypotheses can be removed. 

\begin{itemize}

\item[\textbullet] We need the  hypothesis \ref{hypshen3} because the function $2 \C $ satisfies the three others hypotheses.

\item[\textbullet] The  hypothesis \ref{hypshen4} is necessary because the function $0$ satisfies the three others hypotheses.

\item[\textbullet] The hypothesis \ref{hypshen1} cannot be removed since $\C^{\emptyset'} $ (plain Kolmogorov complexity relativised to the halting problem oracle) satisfies each of three others hypotheses (and clearly differs from the unrelativized version~$\C$).

\item[\textbullet] The hypothesis \ref{hypshen2} cannot be removed because the length function satisfies the three others hypotheses.

\end{itemize}

It could however be the case that hypothesis \ref{hypshen2} be replaced by the weaker ``for all total computable functions~$f$ there exists $c_f$ such that $\A (f(x))\leq \A (x) +c_f $". Our first main result is that this is not case.

\begin{theorem}\label{shenCtot}
There exists a function  $\A:\fs \rightarrow \N $  satisfying hypotheses \ref{hypshen1}, \ref{hypshen3}, \ref{hypshen4} (of Theorem \ref{shenC}) and:
\begin{itemize}

\item[\textbullet] \label{hypshen2tot} For every total computable function $f$ from $\fs $ to $\fs $ there exists a constant $c_f$ such that $\A (f(x))\leq \A (x) +c_f $ for each $x \in \fs$. 

\item[\textbullet] $| \A (x) - \C (x)|$ is not bounded. 

\end{itemize}
\end{theorem}

\begin{proof}
In~\cite{GameInterpretation}, the authors define a notion of total conditional complexity $\overline{\C} (x | y )  $ as the smallest length of a program for a total function~$f$ code such that $f(y) =x$. Of course, $\overline{\C}$ is stable over all total computable functions (i.e. $\overline{\C} (f(x) | y )  \leq  \overline{\C} (x | y )+c_f$ for all total computable functions~$f$) and the authors show that $\overline{\C} $ significantly differs from the plain conditional complexity. However, the function $\overline{\C}$ is not quite suitable for our purposes, for two reasons. First, it is not upper semi-computable and second, its  non-conditional version $\overline{\C} (x|\lambda )$ is equal to $\C$ up to a constant.

In order to construct our counter-example, we first define a way to encode compositions of partial computable functions by a set of strings having the prefix-free property. This encoding is not at all optimal, which is precisely what will make our proof work. We define: $$P = \{ 1^{p_1}0001^{p_2}000\dots 1^{p_k}01 | ~~ \forall k, ~ p_k>0 \}.$$ Notice that $P$ is a prefix-free set. For $\tau= 1^{p_1}0001^{p_2}000\dots 1^{p_k}01 \in P$ we now denote by $ \varphi_\tau$ the function $\varphi_\tau= \psi_{p_1} \circ \psi_{p_2 } \circ \dots \circ \psi_{p_i } $ (recall that~$(\psi_i)$ is a standard enumeration of partial computable functions). We now define a function~$V$, as follows. For all $x \in \fs $ and $\tau \in P$, set

$$ V( \tau x ) =\begin{cases}  \varphi_\tau (x)  & \text { if for all } y \text{ such that }   |y| \leq |x| \text{ we have } \varphi_\tau (y) \downarrow \\
\uparrow & \text{ otherwise }
\end{cases} $$
$P$ is prefix-free, so $V$ is defined without ambiguity and clearly~$V$ is a (partial) computable function. If $\varphi_\tau$ is not a total function, there are only a finite number of strings~$x$ such that $V( \tau x ) \downarrow$. We shall prove that $\A=\C_V$ satisfies the conditions of the theorem. First, $\C_V$ is upper semicomputable and satisfies the counting condition, as it is just the Kolmogorov complexity function associated to the machine~$V$. Moreover, let~$i$ be an index for the identity function (i.e.\ $\psi_i=id$). By definition of $V$, one has $V(1^i01x)=x$, hence $\A(x) \leq |x|+(i+2)$. To see that $\A$ is stable over all total computable functions, let~$f$ be a total computable function and let~$e$ be an index for~$f$. Now, for any string~$x$, let $\tau y$ be such the shortest description of~$x$ for $V$ with $\tau \in P$. By definition of~$V$, this means that $\varphi_\tau(z) \downarrow$ for all $|z| \leq |y|$. And since~$f=\psi_e$ is total, we also know that $\psi_e \circ \varphi_\tau (z) \downarrow$ for all $|z| \leq |y|$. Therefore $\sigma=1^e000\tau$ is a description of~$x$ for~$V$. We have proven that $\C_V(f(x))\leq \C_V(x) + e +3$ for all~$x$. 

It remains to prove that $\C_V$ differs from $\C$, i.e.\ that $\C_V-\C$ takes arbitrarily large values. We prove this by contradiction: Suppose that $|\A (x) -\C (x) |$ is bounded by a constant. For $x \in \fs $, we denote by $\widehat{x} $ the smallest description of~$x$ for~$V$ (by definition this means that $\C_V(x)=|\widehat{x}|$).

Let~$x$ be a string. Let us first write $$\widehat{x} = 1^{p_1}000 1^{p_2}0000\dots 1^{p_k} 01 y.$$ It is easy to see that $$\C (\widehat{x}) \leq  2 \log (p_1) + 2 \log (p_2)+ \dots + 2 \log (p_k) +  2k +|y| +O(1), $$ and since $x$ can be computed from~$\widehat{x}$, this implies a fortiori:  $$\C (x) \leq  2 \log (p_1) + 2 \log (p_2)+ \dots + 2 \log (p_k) +  2k +|y| +O(1). $$
Moreover, by definition of~$V$, $$\C_V(x)=| \widehat{x} | = p_1+ p_2 + \dots + p_k + 3(k-1) +2 + |y| .$$ 
Thus, since we have assumed that $\C_V(x) - \C (x)$ bounded, this shows two things:
\begin{itemize}
\item[\textbullet] the $(p_i)$ appearing in the $\widehat{x}$'s are bounded, and 
\item[\textbullet] the number of $p_i$'s used in each $\widehat{x}$ is bounded. 
\end{itemize}

Formally, we have proven that $\{\tau \in P\, |\,  \exists x \in \fs \text{ such that }  \widehat{x} = \tau y  \}$ is a finite set and that for each $\tau $ in this set, either $\varphi_\tau$ is a total function or for $y $ large enough, $\tau y$ is not in the domain of $V$. Thus, this $\tau$ appears only in a finite number of $\widehat{x}$. 

Finally for $|\widehat{x}|$ large enough (and hence for $|x| $ large enough because ${\{ x | \A (x) \leq n \}}$ is finite for all~$n$), $\widehat{x} =\tau x $ with $\tau \in P$,  and $\varphi_\tau $ is a total computable function. So $$Q= \{ \tau \in P |  \exists^\infty x \in \fs \text{ such that }  \widehat{x} = \tau y  \} . $$ is a finite set of codes of total functions and thus there is only a finite number of $\tau \in P $ in the prefixes of $\widehat{x}$'s. 

Therefore, for $x$ large enough, $\widehat{x}$ is of the form $\tau y $ with $\tau \in Q$ and hence: $$\A (x) = \min \{|\tau y | ~ | ~ \tau \in Q \text{ and } \varphi_{\tau}  (y)=x  \} .$$ Since~$Q$ is finite and all $(\varphi_{\tau})_{\tau \in Q}$ are total, this makes $\A$~computable, contradicting $\A =\C + O(1)$ because no non-trivial lower-bound of $\C $ is computable.

\end{proof}\qed

\section{An axiomatic system for prefix complexity}\label{sectionpourlaprefix}

As we have seen in the last section, there exists a minimal  set of simple properties that characterize plain complexity.               
One may ask whether it is possible to obtain a similar characterization of prefix-free complexity~{\K}.

It is natural to keep the hypotheses \ref{hypshen1} and \ref{hypshen2}, but the other two hypotheses need to be adapted. Indeed, hypothesis~\ref{hypshen3} fails to hold for~$\K$ (i.e.\ $ \K (x)\nleq |x|  +O(1) $), and the sharpest classical upper bound is $\K (x) \leq |x| + \K ( |x|) +O(1)$ (see~\cite{Downey_2010}). 

Accordingly, the hypothesis~\ref{hypshen4} (i.e. $|\{ x | \K (x) \leq n \}|=O(2^n)$) is too weak. The analogue of that counting argument for $\K$ is the classical $$\left| \left\{ x | |x|=n~\text{and}~ \K (x) \leq n +\K(n)-k\right\}\right|=O(2^{n-k}).$$ Another property of~$\K$ that is very often used is $\sum_{x\in \fs} 2^{-\K (x)} < \infty $ (in fact, any upper semi-computable function~$\A$ satisfying $\sum_{x\in \fs} 2^{-\A (x)} < \infty $ is such that $\K \leq \A+O(1)$). Perhaps surprisingly, this set of properties alone is not enough to characterize~$\K $.

\begin{theorem}\label{pourK}
There exists a function $\A$ satisfying the following:

\begin{enumerate}

\item \label{contre_example pour K 1} $\A$ is upper semi-computable.

\item \label{contre_example pour K 2} For every partial computable function $f$ from $\fs $ to $\fs $ there exists a constant $c_f$ such that for each $\A (f(x))\leq \A (x) +c_f $ for each $x \in \fs$.

\item \label{contre_example pour K 3} $\sum_{x\in \fs} 2^{-\A (x)} < \infty $.

\item \label{contre_example pour K 4} $\A (x) \leq |x| + \A ( |x|) +O(1)$ for each $x\in \fs $.

\item \label{contre_example pour K 5} $\left| \left\{ x \in 2^n \left|~\A (x) \leq |n| +\A ( n) - b \right. \right\} \right| \leq O (2^{n-b})$.

\item  $|\A -\K |$ is not bounded. 

\end{enumerate}

\end{theorem}

\rem Since hypothesis \ref{contre_example pour K 3} guarantees the inequality $ \K \leq \A +O(1)$, it would be sufficient, in order to obtain a full characterization of~$\K$, to add the property: ``For every $f$ partial computable prefix-free function there exists $c_f$ such that $\A (f(x)) \leq |x| +c_f$". Indeed, for all $x$ if we denote by $x^*$ a shortest string such that $\mathbb{U}'(x^*)=x$ then $\A (x) = \A ( \mathbb{U}'(x*)) \leq |x^*| +c_{\mathbb{U}'} = \K (x) +c_{\mathbb{U}'}$. However, such a system would not be very satisfactory because it uses the prefix-freeness of functions and thus is mostly a rewording of the definition of~$\K$.

\begin{proof}

We will construct \A by taking $\A= \K + \beta $ with $\beta $ an unbounded function with certain nice properties. The function $\beta $ will be upper semicomputable, non-decreasing, unbounded, such that $$\beta (x) = \beta (|x|) +O(1),$$ and such that for~$f$ partial computable function, there is~$c_f$ such that \begin{equation}\label{robustesse} \beta ( f(n)) \leq \beta(n) +c_f.  \end{equation}  Simple considerations show that $\beta $ has to have a very low growth speed. First let us define Solovay's $\alpha $-function:  
\begin{definition}
The Solovay's $\alpha $-function is defined by:
$$\alpha (n) = \min \{\K (i)| i>n \}. $$

We call order a total, non-decreasing and unbounded (not necessarily computable) function $f: \N \rightarrow \N $. 

\end{definition}

Equivalently $\alpha (n)$ is the length of the shortest string $\tau $ such that $\mathbb{U}'( \tau ) > n$ since $\mathbb{U}'$ is the optimal optimal prefix-free machine chosen to define $\K$. 

$\alpha $ is an order with a very low rate of growth, and actually one can show that it grows more slowly than any computable order. 

\begin{lemma}\label{lemme ordre}
For each $h$ computable order:

\begin{itemize}

\item[(i)] for all $n$, $\alpha (h(n)) = \alpha (n) + O(1)$

\item[(ii)] for all $n$, $\alpha (n) \leq h(n) +O(1)$

\end{itemize} 
\end{lemma}

\begin{proof}
To prove this lemma, we need the following list of trivial facts. 
\begin{itemize}
\item[\textbullet] By the definition of $\alpha$, there exists $j>n$ such that $\alpha (n) = \K (j)$.

\item[\textbullet] There exists $c_{h} $ such that for all $n$, $\K (h(n)) \leq \K (n) +c_{h}$.

\item[\textbullet] $\K (n) \geq \alpha (n)$ for all~$n$.

\item[\textbullet] Since~$h$ and~$\alpha $ are order functions, $h(j) \geq h(n)$ and $\alpha (h(j)) \geq \alpha (h(n))$.
\end{itemize}

Now, one can apply these facts in order to get:
$$\alpha (n) = \K (j) \geq \K (h(j))-c_f \geq \alpha (h(j)) - c_{f} \geq \alpha (h(n)) - c_{f}. $$

To prove $\alpha (n) \leq \alpha (h(n)) +O(1)$, it suffices to consider an inverse order ${\widehat{h} }$ of the order function $h$ defined by: ${\widehat{h}(n) = \max \{ i |h(i) \leq n \} }$. Since ${\widehat{h}}$ is a computable order we have: $$\alpha (n) \leq \alpha (\widehat{f}(h(n))) \leq \alpha (h(n)) + c. $$

To show that $\alpha (n) \leq h(n) +O(1)$, notice that $\K (n ) \leq  n +O(1)$ and so there exists $c$ such that $\alpha (n) \leq n +c$. Finally, by the previous point, we have: $$ \alpha (n) \leq \alpha \left( h(n)\right) + c_h \leq h(n) + d.  $$\qed

\end{proof}

We can show that $\alpha $ satisfies \ref{robustesse} for each total computable function, but there exists some partial computable functions such that $\alpha $ does not satisfy \ref{robustesse}. In the same way we can show that $\K + \alpha $ does not satisfy  condition \ref{contre_example pour K 2} in the statement of the theorem. However, we have a weaker version for partial functions:

\begin{lemma}\label{fun_widjet}
For each  partial computable function $f:\N \rightarrow \N $ there exists $c_f$ such that for all $n$ $$\alpha (\alpha (f(n))) \leq \alpha (n) +c_f .$$

\end{lemma}

\begin{proof}  This follows from the following simple fact. For each $f$ there exists $c_f$ such that:
$$ \alpha (f(n)) \leq \K (f(n)) \leq \K (n) +c_f  \leq n + c_f . $$
 Since $\alpha $ is a sub-linear  order:
$$ \alpha (\alpha (f(n))) \leq \alpha (n + c_f) \leq \alpha (n) +c_f +O(1).$$

\end{proof}\qed

As  stated above, the partial computable functions can  increase too quickly to satisfy the second condition of the theorem. For this reason we introduce a general operator to slow down  sub-linear and upper semi-computable orders:

\begin{definition}[Star-operator]
Let~$f$ be a sub-linear (i.e.\ $f(n)=o (n)$) order function. If we set $$p_f =\max \{n | f(n)\geq n \}  $$ which is well-defined by sub-linearity of~$f$, then, $f^*$ is defined by: $$f^* (n) = \min \{k |f^{(k)}( n)\leq p_f\}.$$

\end{definition}

This operator is a generalization of the so-called $\log^* $, which is precisely the function one gets by taking $f=\log$ in our definition of $f^*$. \\

\rem A simpler definition could be $f^* (n) = \min \{k |f^{(k)}( n)=f^{(k+1)}( n) \}$ but for small values of $n$ and for some function $f$ (for functions with more than one fixed point, for example) this definition is not exactly the same and is in fact less natural.  This star operator will suit our purposes because it possesses some nice properties.

\begin{lemma}\label{star_operator_prop}
Let $f$ a sub-linear order function. The following properties hold: 

\begin{enumerate}

\item \label{prop_star_op 1} $f^*$ is a sub-linear order function.

\item \label{prop_star_op 2} If $f$ is a computable function then so is $f^*$.

\item \label{prop_star_op 3} If $f$ is a upper semi-computable function then so is $f^*$.

\item \label{prop_star_op 4}  $0\leq f^* (n) -  f^*(f^{(i)} (n))  \leq i  $.

%\item \label{prop_star_op 5}  $f^* (n) \leq f^{(i)} (n)+ i $ for all $n$.

\end{enumerate}
\end{lemma}

\begin{proof}

(\ref{prop_star_op 1}) The last claim will ensure sub-linearity. To see that $f^*$ is non-decreasing, if $x \leq y$ then  $f^{(i)}(x) \leq f^{(i)}(y)$ for all $i$ because $f$ is non-decreasing. Finally, $f^*$ is unbounded, for if $f^*$ had a finite limit $d$, then $f^{(d)}$ would be bounded. But this is not possible because $\widehat{f} $ tends to infinity. 

(\ref{prop_star_op 2}) If $f$ is computable then to determine  $f^* (n) $ we can compute the sequence $f^{(1)}(n) , f^{(2)}(n) , \dots , f^{(k)}(n), \dots    $ until we find the first $j$ such that $f^{(j)} (n)\leq p_f$ and we return $j$. 

(\ref{prop_star_op 3}) If $f$ is upper semi-computable then we compute in parallel the approximations of $f^{(k)} (n)$ for all $k$, and we return the least~$k$ such that $f^{(k)} (n) \leq p_f$.

(\ref{prop_star_op 4})  If $f^* (n) \leq i $ then $f^*(f^{(i)}) (n)=0$ because necessarily, $  f^{(i)} (n) \leq p_f $. So $f^* (n) \leq f^*(f^{(i)} (n)) + i$.

 If  $f^* (n) > i $ then $f^* (n) = f^*(f^{(i)}) (n)+i$ by definition of the star-operator. In both cases $f^* (n) \geq f^*(f^{(i)} (n))$.

%(\ref{prop_star_op 5}) For $f$ a sub-linear non-decreasing function, by in the last point it is clear that $f^* (n) \leq f(n) +O(1) \leq n +O(1) $. Therefore: $$f^* (n)= f^*(f^{(i)}) (n) + i  \leq f^{(i)}(n) + i + O(1).$$\qed

\end{proof}

We shall use Solovay's $\alpha$ function transformed by the star-operator. We will show that the function $\A (x) = \K (x) + \alpha^* (x) $ has all the necessary properties to prove the theorem.

By Lemma \ref{star_operator_prop} the function $\alpha^* $ is upper semicomputable, and thus ${\K+\alpha^*}$ is as well.

By Lemma~\ref{fun_widjet} we have that for each  partial computable function $f:\N \rightarrow \N $ there exists $c_f$ such that for all $n$ $$\alpha (\alpha (f(n))) \leq \alpha (n) + c_f$$ and by Lemma~\ref{star_operator_prop} (claim 4), if we apply $\alpha^* $ on each term of the previous inequality we have (since $\alpha^*$ is sub-linear) $$\alpha^* (f(n)) - 2 \leq \alpha^* ( \alpha ( n)+c_f)  \leq \alpha^* (n) +O(1) .$$ This proves the second condition of the theorem.

%By Lemma~\ref{fun_widjet} and Lemma~\ref{star_operator_prop} (claim 4) we have that for each $f:\N \rightarrow \N $ partial computable function there exists $c_f$ such that for all $n$ $$\alpha^* (f(n)) \leq \alpha^* (n) +c_f .$$ This proves the second condition of the theorem.
 
 The property \ref{contre_example pour K 3} is clear because we have $\sum_{x\in \fs} 2^{-\K (x)} < \infty $ and $\A (x) \geq \K (x) $ (because $\alpha^* (x) \geq 0$). 
 
 Finally, since $| .| $ is a computable order function,  by Lemma \ref{lemme ordre} we have $\alpha (x) =  \alpha (|x|)+ O(1) $. By condition \ref{prop_star_op 4} of Lemma \ref{star_operator_prop}, the equality $$\alpha^* (x) =  \alpha^* (|x|)+ O(1) $$ holds. This equality shows that $\A $ satisfies hypotheses \ref{contre_example pour K 4} and \ref{contre_example pour K 5}.

\qed
\end{proof}

It is interesting to notice that the counter-example we produced also invalidates several similar attempts for an axiomatization. 
For example, one could add the condition:  $$\K(xy) \leq \K (x,y) \leq \K(x) + \K(y) +O(1).$$ But $\K + \alpha^* $ also satisfies this. One could then ask whether the more precise inequality $\K (x,y) \leq K(x)+K(y|x)+O(1)$ could help characterizing conditional prefix-free Kolmogorov complexity, but then again, defining $\alpha (x | y) $ by $$\alpha (n | m) = \min \{\K (i |m)| i>n \} $$ and then $\alpha^* ( . | y)$ for each $y$, we get a counter-example by taking~$A(.|.)=K(.|.)+\alpha^*(.|.)$. 

This, together with Theorem \ref{pourK}, shows that the situation is more subtle in the prefix-free complexity case  than in the plain complexity case. Finding a natural characteristic set of properties for \K is left as an open question. 

\section{Acknowledgements}

I would like to express my gratitude to  Laurent Bienvenu without whom this paper would never had
existed. Thanks also to Serge Grigorieff for our numerous discussions during which he helped me progress on this work. 
Finally, thanks to the Chris Porter and three anonymous reviewers for their help in preparing the final version of this paper.

%\nocite{*}
\bibliographystyle{alpha}

\begin{thebibliography}{MMSV10}


\bibitem[Cha66]{papierchaitin}
	Gregory J. Chaitin
\newblock On the Length of Programs for Computing Finite Binary Sequences
	\newblock 1966, J. ACM 13(4): 547-569.

\bibitem[DH10]{Downey_2010}
Rod~G. Downey and Denis Hirschfeldt.
\newblock {\em Algorithmic Randomness and Complexity}.
\newblock Springer, 2010.

\bibitem[Kol65]{kolmo65}
Andre{\"\i}~N. Kolmogorov.
\newblock Three approaches to the definition of the concept ``quantity of
  information''.
\newblock {\em Problemy Pereda\v{c}i Informacii}, pages 3--11, 1965.

\bibitem[MMSV10]{GameInterpretation}
Andrej~A. Muchnik, Ilya Mezhirov, Alexander Shen, and Nikolay Vereshchagin.
\newblock Game interpretation of {K}olmogorov complexity.
\newblock Draft version, 2010.

\bibitem[Nie09]{Nies_2009}
Andr\'e Nies.
\newblock {\em Computability and Randomness}.
\newblock Oxford University Press, 2009.

\bibitem[She82]{shenpapier82}
Alexander Shen.
\newblock Axiomatic description of the entropy notion for finite objects.
\newblock In {\em Logika i metodologija nauki}, Vilnjus, 1982. VIII All-USSR
  conference.
\newblock Paper in Russian.

\bibitem[USV10]{shenbouquin}
Vladimir~A. Uspensky, Alexander Shen, and Nikolai~K. Vereshchagin.
\newblock {K}olmogorov complexity and randomness.
\newblock Book draft version, 2010.



\end{thebibliography}

\newpage

\section*{Appendix}

\subsection{Theorem \ref{shenC}'s full proof}

\begin{theorem}\cite{shenpapier82} \label{shenC}
Let \A be a function of $\fs \rightarrow \N $. If \A   verifies:
\begin{enumerate}

\item \label{hypshen1} \A is computable from above.
\item \label{hypshen2} For every partial computable function $f: \fs \rightarrow \fs $  there exists a constant $c_f$ such that for each $\A (f(x))\leq \A (x) +c_f $ for each $x \in \fs$
\item \label{hypshen3}  $ \A (x)\leq |x |  +O(1) $ for all $x\in \fs$. 
\item \label{hypshen4} $|\{ x | \A (x) \leq n \}|=O(2^n)$

\end{enumerate}

Then $\A (x) = \C (x) +O(1)$

\end{theorem}

This theorem shows that these four properties define exactly  plain Kolmogorov complexity (up to an additive constant, of course). \\

\begin{proof}
First we show $\C (x) \leq \A (x) + O(1) $. We set $S_n =\{ x | \A (x) \leq n \}$. \A is computable from above hence $ S_n $ is a uniformly computably enumerable set.     And by property \ref{hypshen4} there exists $d$ independent of~$n$) such that $| S_n | <2^{n+d-1} $. 

For each $ y $ such that $\A (y)= n$  we can describe $y$ in $S_n $ with a string $\overline{y}$ of length exactly $n+d$, this string represents the rank of $y$ in an enumeration of $S_n$ ($\overline{y}$ is this number with a padding if number does not use  $n+d$ bits). 

Now we describe an algorithm $ E $ to compute $y$ from $\overline{y} $. $E (\overline{y} )  $ is  the $\overline{y}^{\text{th}}$ element in the enumeration of $S_{|\overline{y}| - d}$. To compute $E$ we enumerate the set $\{ \langle n , x \rangle | x\in S_n \}$ and count only elements of with the form $\langle |\overline{y}|-d , x \rangle $ and output $x$ for the $\overline{y}^{\text{th}}$. So we have a machine $E$ such that for all $x$ there exists $p$ such that $E (p) = x$ and $|p|\leq \A (x) + O(1)$. By optimality of $\C$ we have $\C (x) \leq \A (x) + O(1) $.

Now we prove that $ \A (x) \leq \C(x) + O(1) $. If we apply properties \ref{hypshen2} and \ref{hypshen3} in this order we have the next inequalities (we note $x^*$ the shortest string such that $\mathbb{U}(x^* ) = x$): 
$$ \A (x) = \A (\mathbb{U}(x^*)) \leq \A (x^*) +c_U \leq |x^* | + O(1)= \C (x) +O(1).$$
\qed
\end{proof}

\subsection{Conditional complexity}\label{conditionel}

In this section we give a more general axiomatic system for conditional plain complexity.

\begin{theorem} \label{conditionelmarchebien}
Let $ \B : \fs \times \fs \rightarrow \N $. If \B satisfies:

\begin{enumerate}

\item \label{hypshen12} Uniformly in $x$, $y$, $\B (x|y) $ is upper semi-computable.
\item \label{hypshen32} For all $  x, y \in \fs$, $ \B (x|y)\leq |x |  + O(1) $. 
\item \label{hypshen42} For each $  y \in \fs$ we have $|\{ x | \B (x|y) \leq n \}|=O(2^n)$.
\item \label{hypshen52} $\B (\langle x,y \rangle |y)\leq \B (x |y ) + O(1) $.
\item \label{hypshen22} For all $y$ and for every partial computable function $f$ from $\fs $ to $\fs $ there exists a constant $c_f$ such that for each $x \in \fs$: $$\B (f( x  )|y)\leq \B (x |y ) +c_f. $$ 
%\item $ \B (x|x) = O(1) $

\end{enumerate}

Then $|\B (x|y) - \C(x|y) |= O(1)$. 

\end{theorem}

\noindent \textbf{Proof.} \\ 
 For $\C (x|y) \leq  \B (x|y) + c $ the proof is very similar to the unconditional version (theorem \ref{shenC}). In the enumeration of $\{ \langle n , x \rangle |  \B (x|y) < n  \}$ if we know $y$, then we can find $x$ given its index in the set hence $\C (x|y) \leq n +O(1)$. 
 
For the converse. We define $x_y^*$ to be the shortest (and first in lexicographic order) string such that $\mathbb{U} (\langle y, x_y^* \rangle ) = x $ so $\C (x|y) = |x_y^*|$. And if we apply \ref{hypshen22}, \ref{hypshen52} and \ref{hypshen32} (in this order for each inequality) we have: 
 \begin{align*}
    \B (x|y)& = \B (\mathbb{U} (\langle  x_y^*,y \rangle )|y) \leq  \B (\langle  x_y^* ,y\rangle |y)+O(1) \\  & \leq  \B ( x_y^*  |y)  +O(1) \leq   |x_y^*| +O(1) = \C (x|y) +O(1)\\
\end{align*}
\qed

One can ask whether this theorem can be proven with slightly weaker hypotheses. For example we can hope that the hypothesis  $\B (x | x)=O(1)$ instead of hypothesis  \ref{hypshen52} would be sufficient to show the theorem. The next theorem shows that this is not the case. 

\begin{theorem}
There exists a function $\B :\fs \times \fs \rightarrow \fs $ such that $\B $ satisfies the hypotheses \ref{hypshen12}, \ref{hypshen32}, \ref{hypshen42}, \ref{hypshen22} (in theorem \ref{conditionelmarchebien}), $ \B (x|x) = O(1) $ and $|\B (x|y) - \C(x|y) |$ is not bounded.

\end{theorem}

\noindent \textbf{Proof.} \\ 
We will show that  $$\B(x|y) = \min (\C (x), 2\C (x|y))$$ verifies all these properties but of course $|\B (x|y) - \C(x|y) |\not= O(1)$. 
\begin{itemize}
\item[\textbullet] $\B (x|y)$ is upper semi-computable because it is the $\min $ of two functions that are upper semi-computable. 
\item[\textbullet] We know that $\C (x) \leq |x| +O(1)$ so $\B(x|y) \leq |x| +O(1)$.
\item[\textbullet] We have $|\{ x | \C (x|y) \leq n \}|=O(2^n)$ so we have $|\{ x | 2\C (x|y) \leq n \}|=O(2^{n/2})$. And $|\{ x | \C (x) \leq n \}|=O(2^n)$ hence we have $|\{ x | \B (x|y) \leq n \}|=O(2^n)$. 
\item[\textbullet] $ \B (x|x) = O(1) $ comes from $2C(x|y) \leq 2 O(1)= O(1)$.
\item[\textbullet] For $f: \fs \rightarrow \fs $ we must show that there exists $c_f$ such that $\B (f( x  )|y)\leq \B (x |y ) +c_f $. We know already that there exists $c_f$ such that $\C(f(x)) \leq \C(x)+ c_f$ and $2\C (f(x)|y) \leq 2 \C(x|y) +c_f $ so we have: 
\begin{itemize}
\item[\textbullet] If $\B(f(x)|y) = \C(f(x))$ then by definition of $ \B $ (it is a $\min $): $$\B(f(x)|y) = \C(f(x)) \leq 2\C (f(x)|y)  \leq  2\C(x|y) +c_f.$$ In that case we have $\B(f(x)|y) = \C(f(x))\leq \B(x|y)+ O(1) $
\item[\textbullet] In a same way, if $\B(f(x)|y) = 2\C (f(x)|y)$ by definition of $\B $ we have: $$\B(f(x)|y) = 2\C (f(x)|y) \leq \C (f(x))  \leq  \C(x) +c_f. $$ So in that case $\B(f(x)|y) = 2\C (f(x)|y)\leq \B(x|y) +O(1) $ 
\end{itemize}
 So $\B(f(x)|y) \leq \B(x|y) + c_f$ 
\end{itemize}

So $\B$ verifies the hypotheses of the theorem and $|\B (x|y) - \C (x|y)| $ is unbounded.
\qed

\section*{Properties of the Solovay's $\alpha $-function}

In the paper we have shown that for each $h$ computable order:

\begin{itemize}

\item[\textbullet] for all $n$, $\alpha (h(n)) = \alpha (n) + O(1)$

\item[\textbullet] for all $n$, $\alpha (n) \leq h(n) +O(1)$

\end{itemize} 

More generally the proof show that  for each total computable function~$f$: 
$$\alpha (f(n)) \leq \alpha (n) + O(1)$$ because it suffice to consider the total computable order $h_f$ defined by: $$ h_f (n) = \max \{ f(k) | n\geq k \geq 0 \} $$  with the non-decreasing property of $\alpha $ and Lemma \ref{lemme ordre} we have: $$\alpha (f(n)) \leq \alpha (h_f (n)) \leq \alpha (n) +c_{h_f}.$$

As mentioned in the paper this property is not true for partial functions. 

\begin{proposition}[Folklore]

There exists a partial computable function $f: \N \rightarrow \N $ such that $$\forall k\in \N ~~\exists x_k \in \fs ~~ \alpha (f(x_k)) \geq \alpha (x_k) +k .$$

\end{proposition}

\begin{proof} 
We prove this proposition by contradiction. Let $f$ be a partial computable function such that:
$$f(n)= \begin{cases}
       \text{time of the computation of } \mathbb{U}' (n) & \text{ if } \mathbb{U}' (n)\downarrow \\
       \uparrow & otherwise
       \end{cases}$$
where $\mathbb{U}'$ be an optimal machine for \K. 
Let $n_k $ be the integer $i \in [0, 2^k-1]$ such that  the computation time of  $\mathbb{U}' (i)$ is maximal (but finite) among all $i \in \text{Dom} (\mathbb{U}') \cap [0, 2^k -1] $. 

With $m$ such that $m\geq f(n_k)$ we can compute a string $x $ such that $ \K (x ) \geq k $ because we can compute  $\text{Dom}(\mathbb{U}') \cap [0, 2^k -1]$  (because we know an upper bound of the biggest computation time) and take $x \not\in (\text{Dom}(\mathbb{U}') \cap [0, 2^k-1])$, let a function such that $g(m, k) =x$.

Hence for $m\geq f(n_k)$ we have: $$\K (m) + \K (k) +O(1) \geq \K (g(m, k)) \geq \K (x) \geq k $$ and  $$\K (m) \geq \log (n_k) - O(\log(\log (k))).$$ And finally by definition of $\alpha $ we have $\alpha (f(n_k)) \geq \log (n_k) - O(\log(\log (k)))$

So we cannot have $$ \log (n_k) - O(\log(\log (k))) \leq \alpha (f(n_k)) \leq \alpha (n_k) +c_f $$  because $\alpha (n) \leq \log (\log (n))$ (with Lemma \ref{lemme ordre} in the proof of Theorem \ref{pourK}). 

\qed
\end{proof}

With the same idea we can show that $\A = \K + \alpha $ do not verifies $\A (f(x)) \leq \A (x) + c_f$  because for  $n_k$: $$\K (f(n_k)) = \K (n_k)+ O(1).$$

\end{document}